\RequirePackage{fix-cm}
\documentclass[11pt, a4paper, oneside, pagebackref]{amsart}
\usepackage[stretch=10]{microtype}
\usepackage[DIV=11, headinclude=true]{typearea}
\usepackage[hidelinks, linktoc=all]{hyperref}
\usepackage[bb=stix]{mathalpha}
\usepackage{amssymb}

\usepackage[utf8]{inputenc}
\usepackage[T1]{fontenc}
\usepackage{mathtools}

\usepackage{listings}

\usepackage{xcolor}
\definecolor{keywordcolor}{rgb}{0.7, 0.1, 0.1}   
\definecolor{commentcolor}{rgb}{0.4, 0.4, 0.4}   
\definecolor{symbolcolor}{rgb}{0.0, 0.1, 0.6}    
\definecolor{sortcolor}{rgb}{0.1, 0.5, 0.1}      
\definecolor{shadecolor}{rgb}{0.98, 0.98, 0.98}  

\providecommand{\href}[2]{#2}
\providecommand{\texorpdfstring}[2]{#1}
\providecommand*{\backref}{}
\providecommand*{\backrefalt}{}
\renewcommand*{\backref}[1]{}
\renewcommand*{\backrefalt}[4]{%
	\ifcase #1 %
	\or
	  Cited page~#2.
	\else
	  Cited pages~#2.
	\fi
}

\newcommand\MTkillspecial[1]{
  \bgroup
  \catcode`\&=9
  \let\\\relax%
  \scantokens{#1}%
  \egroup
}
\newcommand\DeclarePairedDelimiterMultiline[3]{
  \DeclarePairedDelimiter{#1}{#2}{#3}
  \reDeclarePairedDelimiterInnerWrapper{#1}{star}{
    \mathopen{##1\vphantom{\MTkillspecial{##2}}\kern-\nulldelimiterspace\right.}
    ##2
    \mathclose{\left.\kern-\nulldelimiterspace\vphantom{\MTkillspecial{##2}}##3}}
}

\newcommand{\frakg}{\mathfrak{g}}
\newcommand{\boN}{\mathcal{N}}
\newcommand{\bbk}{\mathbb{k}}
\newcommand{\Q}{\mathbb{Q}}
\newcommand{\Z}{\mathbb{Z}}
\newcommand{\N}{\mathbb{N}}
\newcommand{\R}{\mathbb{R}}
\newcommand{\C}{\mathbb{C}}
\newcommand{\dd}{\mathop{}\!\mathrm{d}}
\DeclarePairedDelimiterMultiline{\abs}{\lvert}{\rvert}
\DeclarePairedDelimiterMultiline{\norm}{\lVert}{\rVert}
\DeclarePairedDelimiterMultiline{\pare}{(}{)}

\DeclareMathOperator{\id}{id}

\renewcommand{\epsilon}{\varepsilon}

\renewcommand{\phi}{\varphi}
\renewcommand{\leq}{\leqslant}

\newcommand{\boL}{\mathcal{L}}

\newtheorem{thm}{Theorem}[section]
\newtheorem{prop}[thm]{Proposition}
\newtheorem{defn}[thm]{Definition}

\newtheorem*{prop*}{Proposition}

\theoremstyle{definition}
\newtheorem{ex}[thm]{Example}
\newtheorem{rmk}[thm]{Remark}
\newtheorem{attempt}[thm]{Attempt}

\numberwithin{equation}{section}

\usepackage{xspace}

\newcommand{\mathlib}{\texttt{mathlib}\xspace}


\title{Higher order differential calculus in mathlib}

\author{Sébastien Gouëzel}
\address{IRMAR, CNRS UMR 6625,
Université de Rennes, 35042 Rennes, France}
\email{sebastien.gouezel@univ-rennes1.fr}

\date{\today}

\begin{document}

\begin{abstract}
We report on the higher-order differential calculus library developed
inside the Lean mathematical library \mathlib. To support a broad range of
applications, we depart in several ways from standard textbook definitions:
we allow arbitrary fields of scalars, we work with functions defined on
domains rather than full spaces, and we integrate analytic functions in the
broader scale of smooth functions. These generalizations introduce
significant challenges, which we address from both the mathematical and the
formalization perspectives.
\end{abstract}

\maketitle

\section{Introduction}

Calculus is a basic and central topic in mathematics, often first studied for
functions from $\R$ to $\R$, then from $\R^n$ to $\R^m$, then for functions
between normed vector spaces -- where the added abstraction and generality
are needed for many important applications. The basic idea of calculus is to
approximate locally a function by a linear map, its differential. Often, this
is not enough, and one should approximate a function by a polynomial or,
equivalently, study iterated differentials, i.e., the differential of the
differential, and so on. This is higher order calculus.

In this article, we describe how higher order calculus has been formalized in
the \mathlib library~\cite{mathlib} for the Lean proof
assistant~\cite{demoura_lean}, insisting on the motivation for several design
choices that differ from mainstream definitions in classical textbooks. In a
nutshell, these choices stem from the broad range of applications that our
formalization should make possible, from Lie groups over $p$-adic fields to
Schwartz spaces for distribution theory.

First-order differential calculus is the study of derivatives of maps. In
this respect, the main textbook definition is the following:
\begin{defn}
\label{defn:derivative} Let $E$ and $F$ be two real normed spaces. A map $f :
E \to F$ is differentiable at a point $x \in E$ if there exists a continuous
linear map $L$ from $E$ to $F$ such that, as $y$ tends to $x$,
\begin{equation*}
  f(y) = f(x) + L (y - x) + o(y - x).
\end{equation*}
\end{defn}
The $o()$ notation means here that $\norm{f(y) - f(x) - L (y-x)} /
\norm{y-x}$ tends to zero as $y$ tends to $x$ while being different from $x$.

When $f$ is differentiable at a point $x$, the linear map $L$ from the
definition is unique. It is the (Fréchet) \emph{differential} (or
\emph{derivative}) of $f$ at $x$, and is often denoted by $Df(x)$ or $D_x f$
or $f'(x)$.

First-order differential calculus has been formalized by Avigad in \mathlib
essentially by following the Isabelle/HOL implementation described
in~\cite[Section~5.2]{analysis_HOL}. While this is not the main focus of this
article, we will recap in Section~\ref{sec:first-order} some features of this
implementation. See also~\cite[Section~4.5]{Affeldt_Cohen_Rouhling_2018} for
a formalization of derivatives in the MathComp Analysis library for Rocq.

For notions of higher-order differentiability, we will follow for instance
the reference textbook~\cite[Section 5.3]{cartan_differential}. Let us denote
by $\boL(E, F)$ the space of continuous linear maps from $E$ to $F$, when $E$
and $F$ are real normed vector spaces. With the operator norm, $\boL(E, F)$
is also a real normed vector space. Informally, a function has $n$ degrees of
smoothness if it has $n$ successive derivatives which are all continuous. We
will say in this case that the function is $C^n$. More formally, the standard
definition of higher-order differentiability is the following inductive
definition.
\begin{defn}
\label{def:Cn} Let $E$ and $F$ be two real normed spaces, and let $f: E \to
F$. We say that $f$ is $C^0$ if it is continuous. For a natural number $n>0$,
we say that $f$ is $C^n$ if it is differentiable and its derivative $Df : E
\to \boL(E, F)$ is $C^{n-1}$. We say that $f$ is $C^\infty$ if it is $C^n$
for all $n \in \N$.
\end{defn}

When $f$ is $C^2$, its second derivative $D^2 f(x)$ is well defined, as an
element of $\boL(E, \boL(E, F))$. There is a canonical (uncurrying)
identification between $\boL(E, \boL(E, F))$ and the space of continuous
bilinear maps from $E \times E$ to $F$, so $D^2 f(x)$ is usually rather seen
as a continuous bilinear map. In the same way, $\boL(E, \boL(E, \dotsb
\boL(E, F)\dotsb)$ is identified with the space of continuous multilinear
maps from $E^n$ to $F$, so the iterated $n$-th derivative of $f$, denoted
$D^n f(x)$, is usually seen as a continuous multilinear map.

\medskip

\mathlib is an integrated mathematics library, whose ultimate goal is to be
able to formalize all current mathematical research. Formalization of notions
in \mathlib should therefore be designed to allow all important applications
down the line. A central application of higher order calculus is the theory
of manifolds. Among the different kinds of manifolds encountered in the
literature, let us mention manifolds with boundary (they are basic blocks in
Thurston's program to classify $3$-dimensional manifolds, solved in the
positive by Perelman and giving Poincaré conjecture in dimension $3$),
manifolds over the $p$-adic numbers (especially important to define $p$-adic
Lie groups), Banach manifolds (key to study groups of diffeomorphisms of
compact real manifolds). For these applications, Definition~\ref{def:Cn} is
not enough. Let us highlight the limitations of this definition.
\begin{enumerate}
\item It is only given over the real numbers. For $p$-adic Lie groups, one
    needs higher-order differential calculus over an arbitrary field.
\item It is only written in the whole space (or, by a straightforward
    generalization, open sets of the vector space). To study manifolds with
    boundary, we need a notion of smooth functions in a closed half-space
    (or in a closed quadrant for manifolds with corners). More generally,
    we should define $C^n$ functions on subsets of vector spaces.
\item \label{issue_analytic} Calculus over normed fields different from
    $\R$ or $\C$ is typically only well behaved for analytic functions (see
    Section~\ref{sec:symm} for more on this). For this reason,
    Bourbaki~\cite{bourbaki_diff} studies manifolds of class $C^n$ for
    $n\in \N \cup\{\infty, \omega\}$ (where $C^\omega$ stands for analytic
    functions), and assumes when the field is not $\R$ or $\C$ that
    $n=\omega$, i.e., that the coordinate changes are analytic. To get a
    unified theory of manifolds over various fields, without code
    duplication, it is therefore necessary to have a single definition of
    $C^n$ functions for $n\in \N \cup\{\infty, \omega\}$.
\end{enumerate}
A suitable definition of $C^n$ functions in \mathlib should solve
simultaneously these three difficulties.

First-order differential calculus has been implemented in \mathlib by Avigad
in 2019, following~\cite{analysis_HOL}. From 2019 to 2025, the author has
brought $C^n$ functions to \mathlib, fixing successively several design flaws
until the current definition, which solves the three issues above. The
definition is now stable. It has proven its robustness, as testified by the
multiple applications notably to manifold theory but also to Fourier analysis
and distributions.

The current paper serves two goals. It is an entry point to people willing to
use $C^n$ functions in \mathlib, describing the main API. It is also a
justification of the design, explaining where more naive approaches would
fail through several mathematical examples highlighting unexpected
subtleties. This can be interesting to mathematicians or to formalizers using
other systems.

\begin{rmk}
There is another standard definition of $C^n$ functions in terms of partial
derivatives $\partial^n f / \partial x_1^{\alpha_1} \dotsm \partial
x_d^{\alpha_d}$ with $\alpha_1 + \dotsb + \alpha_d = n$, for $f : \R^d \to
F$. This definition does not work in infinite dimension, and depends on the
fact that the source space is precisely of the form $\R^d$ (this could be
avoided by specifying instead a finite basis of the source vector space, or
by considering partial derivatives in all directions). These shortcomings
make it worse, both mathematically and from a formalization perspective, than
Definition~\ref{def:Cn}, although it may seem appealing at first as it
requires less algebra (notably, no discussion of spaces of multilinear maps).
This definition would not allow the study of Banach manifolds, or of calculus
in infinite-dimensional Hilbert spaces which is key in quantum mechanics.

Partial derivatives may be recovered from the full iterated derivative.
Denoting by $e_i$ the $i$-th basis vector, then by definition
\begin{equation*}
  \frac{\partial^n f (x)}{\partial x_1^{\alpha_1} \dotsm \partial x_d^{\alpha_d}}
  = D^n f(x)(\underbrace{e_1, \dotsc, e_1}_{\alpha_1\text{ terms}},
  \underbrace{e_2, \dotsc, e_2}_{\alpha_2\text{ terms}}, \dotsc, \underbrace{e_d,\dotsc, e_d}_{\alpha_d\text{ terms}}).
\end{equation*}
We will not use this point of view in this
paper, as it is notationally and combinatorially more cumbersome than the
conceptual approach of Definition~\ref{def:Cn} while applying in more
restricted situations.

Since there are no dependent types in Isabelle/HOL, it is complicated to work
with the space of $n$-multilinear maps from $E^n$ to $F$, for varying $n$, so
Immler and Zhan have to resort to this kind of definition of higher
smoothness, based on partial derivatives, in~\cite{immler_zhang}.
\end{rmk}

\begin{rmk}
\label{rmk_one_dim} There is an even simpler definition of $C^n$ functions
when the source space is the field $\bbk$ itself. In this case, the Fréchet
derivative at a point is a continuous linear map from $\bbk$ to $F$. There is
a canonical identification between $\boL(\bbk, F)$ and $F$, mapping $L$ to
$L(1)$. Through this identification, we recover the elementary notion of the
one-dimensional derivative of a function $f : \bbk \to F$, as a function $f'
: \bbk \to F$. This process can be iterated as the types do not change,
yielding the $n$-th one-dimensional derivative $f^{(n)} : \bbk \to F$. A
function is $C^n$ in this context if all functions $f^{(m)}$ for $m \le n$
are well defined and continuous. In the Mathcomp Analysis Library for Rocq,
the iterated derivatives of functions are currently available in this
one-dimensional setting.
\end{rmk}

\begin{rmk}
There are a lot of other scales of smooth-like functions. For instance, in
$\R^n$, taking into account integrability properties of the function and its
successive derivatives with respect to Lebesgue measure (and/or using Fourier
transform), one can define the Sobolev and Besov spaces. One can also index
the smoothness by a nonnegative real number, by adding Hölder controls on the
iterated derivatives. One can also consider definitions in terms of iterated
differences operators instead of iterated derivatives, which turn out to work
better over ultrametric fields as advocated in~\cite{ultrametric_calculus}
(solving for instance issues such as Proposition~\ref{prop:non_symm} below).
We will not explore these other classes in this paper, and focus only on the
formalization of a class based on Definition~\ref{def:Cn} as this is the most
fundamental in the literature.
\end{rmk}

\section{First-order calculus}
\label{sec:first-order}

As in Isabelle/HOL, the derivatives of functions in \mathlib are defined for
any normed field, on domains. The main predicate is
\begin{lstlisting}
def HasFDerivWithinAt (f : E → F) (f' : E →L[𝕜] F) (s : Set E)
    (x : E) :=
  (fun y ↦ f y - f x - f' (y - x)) =o[𝕜; 𝓝[s] x] (fun y ↦ y - x)
\end{lstlisting}
(The letter \texttt{F} in \texttt{HasFDerivWithinAt} stands for
\emph{Fréchet}, to distinguish it from the one-dimensional derivative
discussed in Remark~\ref{rmk_one_dim}.) In this definition, $\mathbb{k}$ is a
normed field, $E$ and $F$ are two normed vector spaces over $\mathbb{k}$, and
$E \to\!\! L[\bbk] F$ denotes the space of continuous linear maps from $E$ to
$F$, and $\boN [s] x$ is the filter of neighborhoods of $x$ within the set
$s$. This is the direct analogue of Definition~\ref{def:Cn}, except that the
field of scalars does not have to be $\R$, and moreover this definition is
within a domain, i.e., we only require that $f(y) = f(x) + f' (y-x) + o(y-x)$
when $y$ belongs to the set $s$. Using domains is important for the
application to manifolds with boundary, but also for more elementary
statements like the fundamental theorem of calculus, where one considers
derivatives inside intervals.

(The definition is not written exactly like this in \mathlib as it involves
another predicate \texttt{HasFDerivAtFilter}, but it is definitionally equal
to what is written above).

Then we have the predicate \lstinline{DifferentiableWithinAt 𝕜 f s x}
registering that there exists a continuous linear map $f'$ such that
\lstinline{HasFDerivWithinAt f f' s x}. Such a value for the derivative is
recorded in \lstinline{fderivWithin 𝕜 f s x}. This function is total, i.e., it
is also defined when the function is not differentiable, for ease of use. We
give the value $0$ to the derivative in this case by convention, and also
because this choice of junk value ensures that many lemmas remain true even
for non-differentiable functions. Here is the formal definition:
\begin{lstlisting}
def fderivWithin (𝕜) (f : E → F) (s : Set E) (x : E) : E →L[𝕜] F :=
  if HasFDerivWithinAt f (0 : E →L[𝕜] F) s x
    then 0
  else if h : DifferentiableWithinAt 𝕜 f s x
    then Classical.choose h
  else 0
\end{lstlisting}
In general domains (contrary to open sets), the derivative is not unique,
hence the definition resorts to \lstinline{Classical.choose} to choose some
value of the derivative. Using the value $0$ as the derivative whenever
possible, as above, is convenient in several contexts -- for instance, it
ensures that the derivative in a domain of a constant function is zero, while
otherwise the derivative of a constant function within the empty set could be
arbitrary for instance.

\mathlib also contains specializations of these notions when $s$ is the whole
space, called respectively \lstinline{HasFDerivAt f f' x} and
\lstinline{DifferentiableAt 𝕜 f x} and \lstinline{fderiv 𝕜 f x}. It also
contains predicates \lstinline{DifferentiableOn 𝕜 f s} recording that $f$ is
differentiable within $s$ at each point of $s$, and
\lstinline{Differentiable 𝕜 f} recording that $f$ is differentiable
everywhere. A comprehensive API is then given about these notions, including
standard facts such that the differentiability of the sum or the composition
of differentiable functions.

There is a subtlety in domains. It is true that, if $f$ and $g$ have
derivatives respectively $f'$ and $g'$ at $x$ within $s$, then $f + g$ has
derivative $f' + g'$ at $x$ within $s$. In the same way, if $f$ and $g$ are
differentiable at $x$ within $s$, then $f + g$ also is. However, even if $f$
and $g$ are differentiable at $x$ within $s$, it is \emph{not} true in
general that
\begin{lstlisting}
fderivWithin 𝕜 (f + g) s x = fderivWithin 𝕜 f s x + fderivWithin 𝕜 g s x
\end{lstlisting}
because of the lack of uniqueness of derivatives within sets. For instance,
assume that $E = \R^2$ and $F= \R$, let $L$ be the linear map sending $(x_1,
x_2)$ to $x_1$ and $M$ the linear map sending $(x_1, x_2)$ to $x_1+x_2$. Take
$s=\R\times \{0\}$. The function $L$ admits itself as a derivative
everywhere, and in particular along $s$, and the same holds for $M$. However,
as $L$ and $M$ coincide on $s$, the map $L$ also admits $M$ as a derivative
along $s$. It is therefore possible that the functions $f=g=L$ have an
\texttt{fderivWithin} along $s$ equal to $L$, while $f+g=2L$ may have an
\texttt{fderivWithin} along $s$ equal to $2M$, violating the above equality.

It is sometimes useful to ensure that the above equality holds. For this, the
domain should be nice enough to enforce uniqueness of derivatives. This is
true when the tangent directions of $s$ at $x$ span a dense subset of the
whole space. We define a predicate \lstinline{UniqueDiffWithinAt 𝕜 s x}
registering this property, as well as \lstinline{UniqueDiffOn 𝕜 s} registering
that $s$ is a set of unique differentials around each of its points. With
these definitions at hand, the fact that the derivative of the sum is the sum
of the derivatives reads for instance
\begin{lstlisting}
theorem fderivWithin_add {f g : E → F} {x : E} {s : Set E}
    (hxs : UniqueDiffWithinAt 𝕜 s x)
    (hf : DifferentiableWithinAt 𝕜 f s x)
    (hg : DifferentiableWithinAt 𝕜 g s x) :
    fderivWithin 𝕜 (f + g) s x = fderivWithin 𝕜 f s x + fderivWithin 𝕜 g s x
\end{lstlisting}

For whole derivatives (i.e., when $s$ is the whole space), no further
conditions are needed as the whole space (or open sets) are sets of unique
differentials. This corresponds to the uniqueness of derivatives mentioned
after Definition~\ref{defn:derivative}. More generally, convex sets with
nonempty interior are also sets of unique differentials. This applies for
instance to the half-space or the quadrant, which turns out to be important
in manifold theory.

\section{Analytic functions}
\label{sec:analytic}

Our framework for $C^n$ functions has to contain analytic functions, as
explained on Page~\pageref{issue_analytic}. In this paragraph, we describe
our formalization of analytic functions in \mathlib, as a prerequisite to
formalize $C^n$ functions.

A one-dimensional function $f$ is analytic around $0$ if it can be written as
a convergent sum $f(z) = \sum_n c_n z^n$. This generalizes readily: a
function $f:\R^d \to \R$ is analytic at $0$ if it can be written around $0$
as
\begin{equation*}
  f(x_1, \dotsc, x_d)=\sum_\alpha c_\alpha x_1^{\alpha_1} \dotsm x_d^{\alpha_d},
\end{equation*}
where the sum is over all multiindices $\alpha=(\alpha_1, \dotsc, \alpha_d)$
and is required to converge in norm on a neighborhood of $0$ (equivalently,
the coefficients $c_\alpha$ should satisfy a bound $\abs{c_\alpha} \le C
R^{\alpha_1 + \dotsb+\alpha_d}$). However, this definition is restricted to
finite dimension. The correct general definition is the following (see
e.g.~\cite{mujica_infinite_dim_complex_analysis}):
\begin{defn}
\label{defn:analytic} Let $E$ and $F$ be two normed vector spaces over a
normed field $\bbk$. A function $f : E \to F$ is analytic at $x$ if there
exists a sequence $p_n$, where $p_n : E^n \to F$ is a continuous
$n$-multilinear map satisfying $\norm{p_n} \leq C R^n$ for some $C>0$ and
$R<\infty$, such that for all $y$ close enough to $0$ one has
\begin{equation}
\label{eq:analytic_sum_eq}
  f(x + y) = \sum_n p_n(y, \dotsc, y).
\end{equation}
\end{defn}
In one dimension, $p_n(y^{(1)},\dotsc, y^{(n)})$ is just $c_n y^{(1)} \dotsm
y^{(n)}$. In the finite-dimensional situation, the term $p_n(y,\dotsc, y)$ is
the sum of the terms $c_\alpha y_1^{\alpha_1} \dotsm y_d^{\alpha_d}$ over all
multiindices with $\alpha_1 + \dotsb + \alpha_d = n$.

To formalize this definition, we define first the type %
\lstinline{E [×n]→L[𝕜] F} of continuous $n$-multi\-linear maps from $E^n$ to
$F$, with its canonical normed space structure over $\bbk$ (where $E$ and $F$
are normed spaces over $\bbk$). Let %
\lstinline{FormalMultilinearSeries 𝕜 E F} be the space of sequences $p_n$ with
$p_n$ a continuous $n$-multilinear map. We introduce a structure recording
that a function admits a power series expansion as in
Definition~\ref{defn:analytic}, on a ball of radius $r$ intersected with a
set $s$ as we want definitions in domains:
\begin{lstlisting}
structure HasFPowerSeriesWithinOnBall (f : E → F)
    (p : FormalMultilinearSeries 𝕜 E F) (s : Set E)
    (x : E) (r : ℝ≥0∞) : Prop where
  r_le : r ≤ p.radius
  r_pos : 0 < r
  hasSum : ∀ y, x + y ∈ insert x s → y ∈ ball (0 : E) r
    → HasSum (fun n ↦ p n (fun _ : Fin n ↦ y)) (f (x + y))
\end{lstlisting}
In this definition, \lstinline{ℝ≥0∞} is the type of extended nonnegative
reals, i.e., $[0,+\infty]$, and \lstinline{p.radius} is the radius on which
the power series defined by $p$ is converging. The condition \texttt{hasSum}
requires the condition~\eqref{eq:analytic_sum_eq} for $x+y \in s \cup \{x\}$
instead of just $s$ as the theory is very bad otherwise, and in interesting
applications $x$ belongs to $s$ anyway.

Once this predicate is available, we say that $f$ is analytic within $s$ at
$x$ if there exist $p$ and $r>0$ such that $f$ has a power series given by
$p$ on the ball of radius $r$. Formally,
\begin{lstlisting}
def AnalyticWithinAt (𝕜) (f : E → F) (s : Set E) (x : E) :=
  ∃ p : FormalMultilinearSeries 𝕜 E F, ∃ r, HasFPowerSeriesWithinOnBall f p s x r
\end{lstlisting}
Note how this predicate is similar to %
\lstinline{DifferentiableWithinAt 𝕜 f s x}, although it requires a much
stronger smoothness property. In parallel to the differentiability situation,
we introduce further predicates \lstinline{AnalyticOn} recording that a
function is analytic within $s$ at every point of $s$, and
\lstinline{AnalyticAt} recording that a function is analytic at a point
(which corresponds to $s$ being the whole space). We also provide a
comprehensive API for these notions.

The sequence $p_n$ is not uniquely determined by $f$: any sequence $p'_n$
with $p_n(y,\dotsc, y) = p'_n(y,\dotsc, y)$ will also do provided that its
norm does not grow superexponentially.
\begin{ex}
Let $\bbk$ be a normed field and let $f : \bbk^2 \to \bbk$ be given by
$f(x_1, x_2) = x_1 x_2$. Let $p_n = 0$ for $n\ne 2$ and $p_2((x_1, x_2),
(x'_1, x'_2)) = x_1 x'_2$. Then $f(y) = \sum_n p_n (y, \dotsc, y)$, showing
that $f$ is analytic around $0$. The sequence $p'_n$ given by $p'_n=0$ for
$n\ne 2$ and $p'_2((x_1, x_2), (x'_1, x'_2)) = x'_1 x_2$ also satisfies $f(y)
= \sum_n p'_n (y, \dotsc, y)$.

If one wanted a canonical choice for $p_n$, one could try to enforce
symmetry: in this example, one would set
\begin{equation*}
  p_2((x_1, x_2),
  (x'_1, x'_2)) = (x_1 x'_2 + x'_1 x_2) / 2.
\end{equation*}
However, this is not possible if $\bbk$ is of characteristic $2$, for
instance $\bbk = \mathbb{F}_2((t))$. In this case, there is no symmetric
representative, contrary to characteristic $0$ where such a canonical choice
is always possible.
\end{ex}

If $f$ is analytic around a point $x$, and $y$ is small enough, then $f$ is
also analytic around $x+y$ if the target space is complete. In particular,
the set of analyticity points of $f$ is open. This follows from the following
computation:
\begin{align*}
  f(x+y+z) &= \sum_n p_n(y+z, \dotsc, y+z)
  \\& = \sum_n \sum_{I \subseteq \{0,\dotsc, n-1\}} \sum p_n(w^I_1, \dotsc, w^I_n)
\end{align*}
by multilinearity, where $w^I_i$ is equal to $z$ if $i\in I$, and $y$
otherwise. Varying only the coordinates inside $I$, one can see the latter
expression as a $\#I$-multilinear map. Grouping together the terms with the
same number of variables, one gets $f(x+y+z) = \sum q_k(z, \dotsc, z)$ where
$q_k$ is $k$-multilinear. To convert this into a rigorous proof, one should
check the convergence of the series defining $q_k$, which follows from norm
estimates and the completeness of the target space.

The main difficulty here is to write down things precisely, as is often the
case when the paper proof involves a lot of ellipses. The norm controls are
then straightforward albeit tedious.

If $f$ is analytic at $x$ (given by the series of the $p_n$), then $f$ is
also differentiable at $x$, with derivative equal to $p_1$ (modulo the
identification between continuous $1$-multilinear maps and continuous linear
maps). This follows from the series expression and the fact that, for $n \ge
2$, then $p_n(y, \dotsc, y) = O(\norm{y}^n) = o(\norm{y})$ (together with
uniform norm controls). This is formalized as
\begin{lstlisting}
theorem HasFPowerSeriesAt.hasFDerivAt
    (h : HasFPowerSeriesAt f p x) :
    HasFDerivAt f
      (continuousMultilinearCurryFin1 𝕜 E F (p 1)) x

theorem AnalyticAt.differentiableAt
    (h : AnalyticAt 𝕜 f x) :
    DifferentiableAt 𝕜 f x
\end{lstlisting}
Versions within a set are also given in the library.

\begin{ex}
\label{ex:bad_analytic} Here is an example showing that completeness of the
target space is necessary for the openness of the set of analyticity points.
Let $\ell^\infty(\R)$ be the space of bounded real sequences, with the
supremum norm. Define $f : \R \to \ell^\infty(\R)$ as follows: for $\abs{x}<
1$, let $f(x) = (1, x, x^2, \dotsc)$. Otherwise, let $f(x)=0$.

The function $f$ is analytic at $0$. Indeed, for $\abs{x}<1$, we have $f(x) =
\sum c_n x^n$ where $c_n \in \ell^\infty(\R)$ is the sequence with a $1$ at
index $n$ and zeroes elsewhere. As $\ell^\infty(\R)$ is complete, if follows
that $f$ is analytic at any $x$ with $\abs{x}<1$, and therefore
differentiable. Its derivative is $\sum n c_n x^{n-1}$.

Let $F$ be the vector subspace of $\ell^\infty(\R)$ spanned by the $c_n$'s
and the values of $f(x)$ for $x\in \R$. This space is not complete any more.
As $f$ takes its values in $F$, one can define a new function $g : \R \to F$
equal to $f$. The expansion $g(x) = \sum c_n x^n$ is still true in $F$, so
$g$ is analytic at $0$. However, we claim that $g$ is not analytic at any
other point $x$ with $\abs{x}<1$, showing that the set of analyticity points
of $g$ is not open.

Assume by contradiction that $g$ is analytic at $x$. Then $g$ is
differentiable at $x$. Therefore, the derivative $\sum_n n c_n x^{n-1}$
belongs to $F$. By definition of $F$, it is the sum of a finitely supported
sequence, and a finite sum $\sum_i a_i (1, y_i, y_i^2, \dotsc)$. Looking at
the $n$-th coefficient of the sequence, we get $n x^{n-1} = \sum_i a_i y_i^n$
for large $n$. Considering even values of $n$, we can even assume that all
the $y_i$ are nonnegative, replacing $y_i$ by $-y_i$ if necessary. Then the
growth rate of the right hand term is $c z^n$ where $z$ is the largest of the
$y_i$ with a nonzero coefficient. This is a contradiction as this is supposed
to grow like $n x^{n-1}$.

In this example, the function $g$ is only differentiable at $0$. Its
derivative there is $x \mapsto c_1 x$. At other points, it is not
differentiable. With our convention for the derivative, the derivative will
then be zero at these other points. This implies that the derivative of $g$
is not continuous at $0$, albeit $g$ is analytic there!
\end{ex}

Among the standard results on analytic functions that we formalize, we prove
that the composition of analytic functions is analytic, and that the local
inverse of an analytic function is analytic, i.e., the analytic version of
the inverse function theorem. Defining a sequence $q_n$ for the inverse of a
map $f$ corresponding to a sequence $p_n$ is easy formally, but checking that
the growth of $\norm{q_n}$ is not faster than any exponential requires a
further argument. Standard textbooks use complex analysis and the Cauchy
integral formula to control this growth. However, the development of analytic
functions in \mathlib comes \emph{before} complex analysis (which uses a lot
of material on analytic functions). Using it for the proof of the analytic
inverse function theorem would therefore require a lot of intermingling
between the two theories, harming the internal coherence of \mathlib.
Instead, we have devised a direct proof, of a more combinatorial nature and
relying on a tricky inductive estimate originating in~\cite{poeschel}.

All these constructions rely heavily on the finite sums and infinite sums
implementations in \mathlib, as well as results about completeness, which
were already well developed before the start of this project.
%
%

\section{\texorpdfstring{$C^n$}{C\^n} functions}

\subsection{The formal definition}

In this paragraph, we give the definition of $C^n$ functions in \mathlib. As
the definition is more complicated than one might expect naively, we start
with a discussion of several failed attempts to show the difficulties. We
will only talk about $C^n$ functions for $n \le \infty$ at first, and add
analytic functions later on.

\begin{attempt}
As a first attempt, we could follow literally Definition~\ref{def:Cn}: define
inductively the $n$-th iterated derivative of $f$ within the set $s$, as the
derivative within $s$ of the $(n-1)$-th iterated derivative, and say that $f$
is $C^n$ within $s$ if all these iterated derivatives are well defined and
continuous up to stage $n$.

The non-uniqueness of derivatives in domains implies that this strategy
fails: with this definition, one can not prove that continuous linear maps
are $C^n$ on domains, while this should hold for any reasonable definition.

For instance, consider the function $f : (x_1, x_2) \mapsto x_1$ defined on
$\R^2$, and the domain $s = \R \times \{0\}$. Any linear map $\ell_\alpha$ of
the form $(u, v) \mapsto u + \alpha v$ is a derivative of $f$ along $s$, as
only the
horizontal direction matters. It follows that %
\lstinline{fderivWithin ℝ f s x}, as defined above using choice, might be any
function $\ell_{\alpha(x)}$. If $\alpha$ is not continuous on $s$, then $f$
would not even be $C^1$ on $s$.
\end{attempt}

This failed attempt shows that a reasonable definition should not be in terms
of the specific derivative %
\lstinline{fderivWithin ℝ f s x} (which might be badly behaved) but instead in
terms of the existence of a nice sequence of derivatives. Using the type %
\lstinline{FormalMultilinearSeries 𝕜 E F} of sequences $p_n$ of
$n$-multi\-linear maps, we record in %
\lstinline{HasFTaylorSeriesUpToOn n f p s} the fact that $p_0$ coincides with
$f$ on $s$, and $p_{m+1}$ is the derivative of $p_m$ for any $m<n$, and $p_n$
is continuous. This would make it possible to say that a function is $C^n$ if
it admits a Taylor series of order $n$ (except that in reality we will again
need to adapt this definition).

The formalized definition is the following:
\begin{lstlisting}
structure HasFTaylorSeriesUpToOn
    (n : WithTop ℕ∞) (f : E → F)
    (p : E → FormalMultilinearSeries 𝕜 E F) (s : Set E) : Prop where
  zero_eq : ∀ x ∈ s, (p x 0).curry0 = f x
  fderivWithin : ∀ m : ℕ, m < n → ∀ x ∈ s,
    HasFDerivWithinAt (p · m) (p x m.succ).curryLeft s x
  cont : ∀ m : ℕ, m ≤ n → ContinuousOn (p · m) s
\end{lstlisting}
In this definition, \lstinline{WithTop ℕ∞} is the set $\N$ with two added
elements $\infty$ and $\omega$.  We include these two added elements to
anticipate over the definition of $C^n$ for $n \in \N \cup \{\infty,
\omega\}$. The \lstinline{F} in the name stands for ``Fréchet'', to emphasize
that this is not a one-dimensional notion. The property \lstinline{zero_eq}
essentially requires that $f$ and $p_0(x)$ coincide on $s$. However, $f (x)$
belongs to $F$ while $p_0(x)$ is a function from $E^0$ to $F$, so they can
not coincide formally: to get a sensible statement, one needs to add the
canonical identification between $F$ and $E^0 \to F$, which we may see as a
curryfication step. In the same way, the property \lstinline{fderivWithin}
requires essentially that $p_{m+1}(x)$ is a derivative of $p_m$ at $x$,
except that $p_{m+1}$ is a multilinear map $E^{m+1} \to F$ while the
derivative of $p_m$ is a linear map from $E$ to multilinear maps from $E^m$
to $F$. Again, these two sets of functions are canonically identified through
a curryfication step, written as \lstinline{curryLeft} in the definition.
Finally, we should also require the continuity of $p_n$, except that it does
not make sense for $n=\infty$ or $\omega$. Therefore, \lstinline{cont}
requires instead continuity of all $p_m$ for $m\le n$, although the new
information is only for the last term if there is one: by
\lstinline{fderivWithin} we already know that the other terms are
differentiable, hence continuous.

\begin{rmk}
Formal multilinear series, i.e., sequences $(p_n)$ of continuous
$n$-multilinear maps, show up both in the definition of analytic functions as
$f(x+y) = \sum p_n(y, \dotsc, y)$ in Definition~\ref{defn:analytic}, and in
sequences of iterated derivatives in Taylor series. These two uses of the
same formal object do \emph{not} correspond to each other. For instance,
consider the real function $f(x) = 1/(1-x)$ around $0$. It is a analytic at
$0$, with expansion $f(x) = \sum x^n$, so $p_n(y, \dotsc, y) = y^n$. The
$n$-th derivative of $f$ at $0$, however, is $D^nf(0)(v_1, \dotsc, v_n) = n!
v_1 \dotsm v_n$ (corresponding to the statement in terms of $1$-dimensional
derivatives that $f^{(n)}(0) = n!$). There is a discrepancy of $n!$ between
the analytic function point of view and the derivatives point of view. This
is not surprising given the standard expansion of a one-dimensional real
analytic function around $0$,
\begin{equation*}
  f(y) = \sum \frac{f^{(n)}(0)}{n!} y^n.
\end{equation*}
The coefficient $n!$ here has a nice combinatorial interpretation as the
cardinality of the group of permutations of $\{1,\dotsc, n\}$,
see~\eqref{eq:deriv_analytic}.
\end{rmk}

\begin{attempt}
\label{att:locality} As a second attempt, we could say that a function is
$C^n$ (for $n\in \N \cup \{\infty\}$) on a set $s$ if there exists a sequence
$(p_m)$ which is a Taylor series for $f$ up to $n$ on $s$ in the above sense.

The issue with this definition is locality, both in space and in regularity.
A good notion of $C^n$ function should satisfy the following: assume that,
for all $x \in s$, there is a neighborhood $u$ of $x$ in $s$ on which $f$ is
$C^n$. Then $f$ is $C^n$ on $s$. Also, if a function is $C^n$ on $s$ for all
natural number $n$, then it should be $C^\infty$. The above definition a
priori fails these two properties, as we will explain now.s

If $f$ is $C^n$ in this sense on a neighborhood $u$ of $x$ in $s$, one gets a
corresponding sequence of multilinear maps $p_m^{(u)}$, indexed by $u$,
giving a Taylor series for $f$ on $u$. If these were compatible in the sense
that $p_m^{(u)} = p_m^{(v)}$ on $u \cap v$, one could glue them together for
$f$ on the whole $s$. However, due to the lack of uniqueness of derivatives,
there is no guarantee of compatibility. In a finite-dimensional real vector
space, one could instead extract a locally finite covering of the $u$'s, then
use a smooth partition of unity $\rho_u$ subordinate to this locally finite
covering, and form the (locally finite) sum $\sum_u \rho_u p_m^{(u)}$. This
would be a Taylor series for $f$ on the whole set $s$. However, the existence
of smooth partitions of unity is not always true in infinite dimension, or on
other fields. For instance, there is no complex-differentiable partition of
unity on $\C$ as a complex-differentiable function with compact support
vanishes (as it is analytic).
\end{attempt}

Our solution to this issue is to bake in locality into the definition: we
will say that $f$ is $C^n$ within $s$ at $x$ (for $n \in \N \cup \{\infty\}$)
if, for every natural number $n' \le n$, there exist a neighborhood $u$ of
$x$ within $s$ and a formal multilinear series $(p_m)$ such that $f$ admits
$(p_m)$ as a Taylor series on $v$ up to order $n'$. Then, we will say that
$f$ is $C^n$ on $s$ if it is $C^n$ within $s$ at every point of $s$. By
design, this definition is local both in space and in regularity.

\medskip

Let us now integrate analytic functions into our definition of $C^n$
functions.

\begin{attempt}
As a third attempt, let us say that $f$ is $C^n$ within $s$ at $x$ (for $n
\in \N \cup \{\infty\}$) using the above local definition, and that it is
$C^\omega$ within $s$ at $x$ if it is analytic within $s$ at $x$.

With this definition, however, it is not true that if $f$ is $C^n$ and $m\le
n$ then $f$ is $C^m$, while this is a property we definitely expect of a
scale of smooth functions. This is due to the bad behavior of analytic
functions on non-complete spaces. For instance, the function $g$ of
Example~\ref{ex:bad_analytic} is analytic at $0$ but it is not $C^1$ at $0$,
as explained at the end of the example.
\end{attempt}

A solution would be to require all our target spaces to be complete, as is
for instance done in~\cite{bourbaki_diff}. However, this is unnecessarily
restrictive: most of the theory can be built without completeness, and
completeness should only be brought as an assumption in theorems that really
need it. Instead, our strategy is to require that $C^\omega$ functions are
analytic, as well as all their derivatives. In this way, we recover
monotonicity of the scale of smooth functions all the way up to $\omega$.

We can now give the formalized definition of $C^n$ functions used in
\mathlib. We introduce a predicate %
\lstinline{ContDiffWithinAt 𝕜 n f s x} registering that $f$ is $C^n$ (over the
field $\bbk$) within $s$ at $x$, defined as follows:
\begin{lstlisting}
def ContDiffWithinAt
    (𝕜) (n : WithTop ℕ∞) (f : E → F) (s : Set E) (x : E) : Prop :=
  match n with
  | (n : ℕ∞) => ∀ m : ℕ, m ≤ n → ∃ u ∈ 𝓝[insert x s] x,
      ∃ p : E → FormalMultilinearSeries 𝕜 E F, HasFTaylorSeriesUpToOn m f p u
  | ω => ∃ u ∈ 𝓝[insert x s] x,
      ∃ p : E → FormalMultilinearSeries 𝕜 E F,
      HasFTaylorSeriesUpToOn ω f p u ∧
      ∀ i, AnalyticOn 𝕜 (fun x ↦ p x i) u
\end{lstlisting}
If the regularity $n$ belongs to $\N \cup \{\infty\}$, then for all integer
$m \le n$ we require the existence of a neighborhood $u$ of $x$ within $s
\cup \{x\}$ on which $f$ admits a Taylor series up to order $m$. This
corresponds to the above discussion, except that we use $s \cup \{x\}$
instead of $s$: this avoids pathologies where the Taylor series would be
discontinuous at $x$. For $n=\omega$, we require the existence of a
neighborhood $u$ of $x$ within $s$ on which $f$ admits a Taylor series at
infinite order, made of analytic functions.

Once we have this predicate, we define a version not localized to a set,
called \lstinline{ContDiffAt} (obtained by taking for $s$ the whole space in
the above definition), a version asking that $f$ is $C^n$ within $s$ at all
points of $s$ called \lstinline{ContDiffOn}, and a version in the whole space
called \lstinline{ContDiff}. These are the perfect analogues of the
corresponding definitions for differentiable functions and for analytic
functions: such a coherence is very helpful to the users of the library. The
lemmas for these three notions also follow a completely identical naming
scheme.

\medskip

Let us now give a few basic statements showing that the definition above
lives up to our promises.

The scale of $C^n$ functions is monotone with respect to $n$:
\begin{lstlisting}
theorem ContDiffWithinAt.of_le
    (h : ContDiffWithinAt 𝕜 n f s x) (hmn : m ≤ n) :
    ContDiffWithinAt 𝕜 m f s x
\end{lstlisting}
In particular, a $C^\omega$ function is $C^\infty$.

In a complete space, $C^\omega$ and analytic coincide:
\begin{lstlisting}
theorem contDiffWithinAt_omega_iff_analyticWithinAt
    [CompleteSpace F] :
    ContDiffWithinAt 𝕜 ω f s x ↔ AnalyticWithinAt 𝕜 f s x
\end{lstlisting}

For $n\ne \infty$, a function is $C^n$ within a set at a point if and only if
there exists a neighborhood on which it has a Taylor series at order $n$, and
moreover if $n=\omega$ all terms in this Taylor series should be analytic:
\begin{lstlisting}
lemma contDiffWithinAt_iff_of_ne_infty (hn : n ≠ ∞) :
    ContDiffWithinAt 𝕜 n f s x ↔ ∃ u ∈ 𝓝[insert x s] x,
      ∃ p : E → FormalMultilinearSeries 𝕜 E F, HasFTaylorSeriesUpToOn n f p u ∧
        (n = ω → ∀ i, AnalyticOn 𝕜 (fun x ↦ p x i) u)
\end{lstlisting}
Note that this is not true for $n=\infty$: consider $f : \R \to \R$ which is
$C^n$ and not better on $(-1/n, 1/n)$. Then it is $C^\infty$ at $0$, but not
$C^\infty$ on a neighborhood of $0$.

For $n\ne \infty$, a function is $C^{n + 1}$ on a domain iff locally, it has
a derivative which is $C^n$ (and moreover the function is analytic when $n =
\omega$).
\begin{lstlisting}
theorem contDiffWithinAt_succ_iff_hasFDerivWithinAt
    (hn : n ≠ ∞) :
    ContDiffWithinAt 𝕜 (n + 1) f s x ↔ ∃ u ∈ 𝓝[insert x s] x,
      (n = ω → AnalyticOn 𝕜 f u) ∧ ∃ f' : E → E →L[𝕜] F,
      (∀ x ∈ u, HasFDerivWithinAt f (f' x) u x) ∧ ContDiffWithinAt 𝕜 n f' u x
\end{lstlisting}

\subsection{Operations on functions respect smoothness}

There is a tension between our definition and standard proofs of properties
of $C^n$ functions. Here is the standard proof that the sum of two $C^n$
functions is $C^n$, based on Definition~\ref{def:Cn}.
\begin{prop}
\label{prop:add_Cn} Let $E, F$ be two real normed spaces and $f, g : E \to F$
be two $C^n$ functions for $n \in \N$. Then $f+g$ is $C^n$.
\end{prop}
\begin{proof}
We argue by induction on $n$. For $n=0$, the result is standard. For $n>0$,
according to Definition~\ref{def:Cn}, we should check that $f+g$ is
differentiable, which is standard, and that $D(f+g)$ is $C^{n-1}$. As $D(f+g)
= Df + Dg$, this follows from the inductive assumption applied to the
functions $Df$ and $Dg$ from $E$ to $\boL(E, F)$.
\end{proof}
Our definition of $C^n$ functions in \mathlib is not recursive, so the above
proof can not be adapted directly. However, we have another recursive
characterization stated above in
\lstinline{contDiffWithinAt_succ_iff_hasFDerivWithinAt}, that we could try to
use. This approach has two issues.

First, it does not work for $C^\omega$ functions: a function is $C^\omega$ if
it is analytic and has a derivative which is $C^\omega$ -- this can not be
used in an inductive-like argument. Second, there are universe polymorphism
issues. The proof above applies the inductive assumption to another pair of
functions, $Df$ and $Dg$, with a different target space, $\boL(E, F)$.
Therefore, a proper formulation of the inductive assumption would be: for all
real normed spaces $E$ and $F$ and all pairs of $C^n$ functions $f, g$ from
$E$ to $F$, the sum $f+g$ is $C^n$. To avoid Russell-like paradoxes, in the
dependent type theory of Lean, spaces live in a \emph{universe}, which can
not be quantified in statements or proofs: one can not write ``for all
universes $u$, for all spaces in the universe $u$, then ...'' In particular,
for an induction to make sense, the spaces featured in the induction should
live in a fixed universe. If $E$ is in a universe $u$ and $F$ is in a
universe $v$, then $\boL(E, F)$ lives in another universe $\max(u,v)$, which
means that our induction will not work correctly. This issue can be worked
around as follows. First prove a less polymorphic statement, in which $E$ and
$F$ belong to the same universe $u$ (as well as all the spaces of multilinear
maps that show up in the proof, enabling the induction). Then reduce the
general case to the specific case: lift $E$ and $F$ to spaces $E'$ and $F'$
in a common universe, together with $C^n$ lifted functions $f'$ and $g'$,
then argue that $f' + g'$ is $C^n$ using the less polymorphic statement, and
finally deduce that $f + g$ is $C^n$ as being $C^n$ is invariant under
continuous linear isomorphisms.

Given these two difficulties, and especially the first one which can not be
worked around, it follows that the formalized proof should probably not
follow the inductive classical proof above, and instead be more in line with
the definition choice we have made. Here is another paper proof of
Proposition~\ref{prop:add_Cn}, but closer to our formal definition.

\begin{proof}[Another proof of Proposition~\ref{prop:add_Cn}]
As $f$ and $g$ are $C^n$, they admit respective Taylor series $p$ and $q$ up
to order $n$. Then one checks readily that $p+q$ is a Taylor series up to
order $n$ for $f+g$, proving that $f+g$ is $C^n$.
\end{proof}
There is no inductive character to this proof, which means that it is
straightforward to formalize, to get the following:
\begin{lstlisting}
theorem ContDiffWithinAt.add
    (hf : ContDiffWithinAt 𝕜 n f s x)
    (hg : ContDiffWithinAt 𝕜 n g s x) :
    ContDiffWithinAt 𝕜 n (fun x ↦ f x + g x) s x
\end{lstlisting}

A lesson here is that, instead of proving inductively that functions are
$C^n$, it works better with our definition to \emph{give a formula} for the
successive derivatives, as an operation on the Taylor series of the original
functions. For the sum of functions, this is very easy (the Taylor series of
the sum is the sum of the Taylor series), but this creates a difficulty for
the composition.

\begin{thm}
Let $E, F, G$ be real normed spaces. Consider $f : E \to F$ and $g: F \to G$
two $C^n$ functions. Then the composition $g \circ f$ is $C^n$.
\end{thm}
Here is the standard textbook inductive proof.
\begin{proof}
The statement is easy for $n=0$. For $n>0$, the function $g \circ f$ is
differentiable as the composition of differentiable functions, and its
derivative is given by $D(g \circ f)(x) (v) = Dg (f(x)) (Df(x) (v))$. We
should show that this is a $C^{n-1}$ function of $x$.

Let $M : \boL(E, F) \times \boL(F, G) \to \boL(E, G)$ be the composition of
linear maps. This is a continuous bilinear map, hence $C^\infty$. Let $\Phi :
E \to \boL(E, F) \times \boL(F, G)$ be $(Df, (Dg) \circ f)$. As $f$ is $C^n$,
then $Df$ is $C^{n-1}$. Moreover, $(Dg) \circ f$ is $C^{n-1}$ using our
inductive assumption, as the composition of two $C^{n-1}$ functions. It
follows that $\Phi$ is $C^{n-1}$. By construction, $D(g\circ f) = M \circ
\Phi$. It is the composition of two $C^{n-1}$ functions, therefore $C^{n-1}$
again by our inductive assumption.
\end{proof}

It is not obvious from this proof how to give a formula for the derivative of
$g \circ f$, and the formula would probably be very complicated. For
instance, the second derivative is
\begin{align*}
  D^2 (g\circ f) (x) (v_0, v_1) = {}& D^2 g(f(x)) (Df(x) v_0, Df(x) v_1)
  \\ & + Dg (f(x)) (D^2 f(x) (v_0, v_1)).
\end{align*}

In a previous \mathlib implementation of $C^n$ functions, in which analytic
functions were not yet integrated in the smooth hierarchy, the fact that the
composition of $C^n$ functions is $C^n$ was proved following the above
inductive argument (and dealing with the universe polymorphism issue thanks
to a lifting step, as explained above). This made it possible to avoid the
complicated formula for the $n$-th derivative of a composition.

In the refactor integrating analytic functions in the smooth hierarchy, the
smoothness of composition of functions has been the most difficult point: we
had to switch proofs since the inductive argument does not make sense for
$C^\omega$ functions. Instead, we state and prove the formula for the Taylor
series of a composition, called the Faà di Bruno formula. This formula is
usually given for functions of one variable, or using partial derivatives for
functions of finitely many variables, but there is a general formula in terms
of Fréchet derivatives which we will now describe. It is expressed in terms
of partitions $I$ of $\{0, ..., n-1\}$. For such a partition, denote by $k$
its number of parts, write the parts as $I_0, \dotsc, I_{k-1}$ ordered so
that $\max I_0 < ... < \max I_{k-1}$, and let $i_m$ be the number of elements
of $I_m$. Then
\begin{equation*}
D^n (g \circ f) (x) (v_0, \dotsc, v_{n-1}) = \sum_I
    D^k g (f (x)) (D^{i_0} f (x) (v_{I_0}), ..., D^{i_{k-1}} f (x) (v_{I_{k-1}}))
\end{equation*}
where the sum is over all partitions of $\{0, \dotsc, n-1\}$ and by $v_{I_m}$
we mean the family of vectors $(v_i)$ with indices $i$ in $I_m$.

The formula is straightforward to prove by induction. Differentiating
\begin{equation*}
  D^k g (f (x)) (D^{i_0} f (x) (v_{I_0}), \dotsc, D^{i_{k-1}} f (x)
(v_{I_{k-1}})),
\end{equation*}
we get a sum with $k + 1$ terms where one differentiates
either $D^k g (f (x))$ or one of the $D^{i_m} f (x)$, amounting to adding to
the partition $I$ either a new atom $\{-1\}$ to its left, or extending $I_m$
by adding $-1$ to it. In this way, one obtains bijectively all partitions of
$\{-1, \dotsc, n-1\}$, and the proof can go on (up to relabeling).

The formalization of this formula is rather intricate. Even writing down
things precisely is hard, as indicated by the number of ellipses in the above
formula. The most technical part is to show that extending partitions of
$\{0,\dotsc, n-1\}$ by the processes described above, one obtains bijectively
all partitions of $\{-1, \dotsc, n-1\}$. Overall, the formalization of the
Faà di Bruno formula takes slightly more than 1000 lines.

Once this formula is available, the proof that the composition of $C^n$
functions is $C^n$ is easy, as the formula tells how to construct a Taylor
series for $g \circ f$ in terms of Taylor series for $g$ and $f$. For
$C^\omega$ functions, one should additionally use that the composition of
analytic functions is analytic, a fact which has already been discussed in
Section~\ref{sec:analytic}.

\medskip

Closely related to composition is inversion. The inverse function theorem
states that, if a function between complete spaces is $C^n$ at a point $x_0$
with invertible derivative, then it admits a local inverse which is also
$C^n$. There are two parts in this theorem. First, if $f$ is $C^1$ at $x_0$
with invertible derivative, then it admits a local inverse $g$. Second, if in
fact $f$ is $C^n$, then $g$ is also automatically $C^n$. The two parts are
essentially independent.

The deeper one is the first one, constructing the inverse by a suitable
application of the Banach contraction principle. The $C^1$ assumption is not
the correct one over a general field: one should instead require that $f$ is
strictly differentiable at $x_0$, i.e., $f(y)-f(z) - Df(x_0) (y-z) = o(y-z)$
as $y$ and $z$ tend to $x_0$. This follows from a $C^1$ assumption over $\R$
or $\C$ (by the mean value inequality), but not over general fields. However,
strict differentiability follows from analyticity, which is the right
smoothness condition to consider over general fields anyway as further
discussed in Section~\ref{sec:symm}.

For the second part, one can check the regularity inductively: as $f \circ g
= \id$, one gets $Df(g(x)) Dg(x) = \id$, so $Dg(x) = Df(g(x))^{-1}$. This
formula makes it possible to transfer a $C^n$ assumption on $g$ to a $C^n$
assumption on $Dg$, deducing that $g$ is $C^{n+1}$ and proceeding
inductively. Note that this argument deduces that $g$ is $C^{n+1}$ from the
fact that it is $C^n$. The spaces or the maps do not change over the
induction, so there is no formalization difficulty here. The analytic case
has to be treated separately as one needs additional growth controls, but
this has already been discussed in Section~\ref{sec:analytic}. Altogether, we
get a generic statement of the inverse function theorem, working both for
$n\ne\omega$ and $n=\omega$ in a unified way.

The first step is formalized by constructing a local inverse from a strictly
differentiable map with invertible derivative. This is phrased in \mathlib
with \lstinline{PartialHomeomorph}, which contains the data of a map (here
$f$) and another map, which are locally inverse to each other.
\begin{lstlisting}
def HasStrictFDerivAt.toPartialHomeomorph
    (f : E → F) {f' : E ≃L[𝕜] F} [CompleteSpace E]
    (hf : HasStrictFDerivAt f f' a) :
    PartialHomeomorph E F
\end{lstlisting}
Here, \lstinline{E ≃L[𝕜] F} is the space of continuous linear bijections
between $E$ and $F$. Therefore, our assumption \lstinline{hf} is a way of
saying that the derivative of $f$ is invertible.

The second step is formalized by saying that if two functions are locally
inverse to each other, then the second inherits the regularity properties of
the first:
\begin{lstlisting}
theorem PartialHomeomorph.contDiffAt_symm
    [CompleteSpace E] (f : PartialHomeomorph E F)
    {f₀' : E ≃L[𝕜] F} {a : F} (ha : a ∈ f.target)
    (hf₀' : HasFDerivAt f f₀' (f.symm a))
    (hf : ContDiffAt 𝕜 n f (f.symm a)) :
    ContDiffAt 𝕜 n f.symm a
\end{lstlisting}
This can be applied to the partial homeomorphism %
\lstinline{HasStrictFDerivAt.toPartialHomeomorph f hf} that has been
constructed in the first step.

\subsection{Iterated derivatives}

To circumvent the difficulties coming from the non-unique\-ness of
derivatives, we have defined $C^n$ functions via the existence of a nice
Taylor series. However, as for \lstinline{fderivWithin} and
\lstinline{fderiv}, it is often useful to have available \emph{some} choice
of an iterated derivative -- although in general one will only be able to
prove that it behaves well under a unique differentials assumption on the
domain. There are two possible choices here: one may define the $(n + 1)$-th
iterated derivative either as the derivative of the $n$-th derivative, or as
the $n$-th derivative of the derivative. For type-theoretic reasons, the
$n$-th derivative of the derivative can not really be used: if it were an
inductive definition over $n$, it would also need to generalize the spaces as
to get the $n+1$-th derivative of $f:E\to F$ one would need the $n$-th
derivative of $Df : E \to \boL(E, F)$. Since the universe of $\boL(E,F)$ is
usually not the same as that of $F$, this can not be done in a plain
inductive definition in Lean. Therefore, we define the $(n+1)$-th derivative
as the derivative of the $n$-th derivative (modulo currying). This is in line
with our definition of $C^n$ functions using sequences of multilinear maps
$p_n$, where $p_{n+1}$ has to be a derivative of $p_n$.

The formal definition is the following:
\begin{lstlisting}
def iteratedFDerivWithin (𝕜) (n : ℕ) (f : E → F) (s : Set E) :
    E → E[×n]→L[𝕜] F :=
  Nat.recOn n
    (fun x ↦ ContinuousMultilinearMap.uncurry0 𝕜 E (f x))
    (fun _ rec x ↦ ContinuousLinearMap.uncurryLeft
      (fderivWithin 𝕜 rec s x))
\end{lstlisting}
This definition is written using the bare recursor \lstinline{Nat.recOn}
instead of versions with more syntactic sugar because of the position of $n$
in the arguments and of difficulties of the elaborator with the complicated
dependent type \lstinline{E[×n]→L[𝕜] F}. Its content is however what we
expect: for $n=0$, use the function seen as a $0$-multilinear map, and for
$n+1$ use the derivative of the $n$-th derivative, seen as an
$(n+1)$-multilinear map.

We also define \lstinline{iteratedFDeriv}, corresponding to taking for $s$
the whole space.

In domains with unique differentials (including the whole space, open sets,
or convex sets with nonempty interior), if a function has a Taylor series up
to order $n$ then this Taylor series has to coincide with
\lstinline{iteratedFDerivWithin}, by uniqueness. Therefore, our proofs that
the sum or composition of $C^n$ functions are $C^n$, giving formulas for the
Taylor series, yield also formulas for iterated derivatives. For instance,
for the addition, we get
\begin{lstlisting}
theorem iteratedFDerivWithin_add_apply
    (hf : ContDiffWithinAt 𝕜 i f s x)
    (hg : ContDiffWithinAt 𝕜 i g s x)
    (hu : UniqueDiffOn 𝕜 s) (hx : x ∈ s) :
    iteratedFDerivWithin 𝕜 i (f + g) s x =
      iteratedFDerivWithin 𝕜 i f s x
      + iteratedFDerivWithin 𝕜 i g s x
\end{lstlisting}

When we defined $C^n$ functions, we had to take care of locality to avoid the
non-uniqueness of derivatives, as discussed in Attempt~\ref{att:locality}. In
domains with unique differentials, the issue disappears: different Taylor
series around different points have to coincide locally, and can therefore be
glued together. In other words, they coincide with a global object, which is
\lstinline{iteratedFDerivWithin}. Let \lstinline{ftaylorSeriesWithin 𝕜 f s} be
the Taylor series formed using the sequence of iterated derivatives. We get
\begin{lstlisting}
theorem ContDiffOn.ftaylorSeriesWithin
    (h : ContDiffOn 𝕜 n f s) (hs : UniqueDiffOn 𝕜 s) :
    HasFTaylorSeriesUpToOn n f
      (ftaylorSeriesWithin 𝕜 f s) s
\end{lstlisting}

We have mentioned above that for the inductive definition of
\lstinline{iteratedFDerivWithin}, we had to write it as the derivative of the
iterated derivative, for type theoretic reasons. However, once the definition
is available, one can prove that it coincides with the other natural
definition, i.e., it is also the iterated derivative of the derivative
(modulo currying):
\begin{lstlisting}
theorem iteratedFDerivWithin_succ_eq_comp_right
    (hs : UniqueDiffOn 𝕜 s) (hx : x ∈ s) :
    iteratedFDerivWithin 𝕜 (n + 1) f s x =
    (continuousMultilinearCurryRightEquiv' 𝕜 n E F).symm
      (iteratedFDerivWithin 𝕜 n (fderivWithin 𝕜 f s) s x)
\end{lstlisting}
This requires the domain to have unique differentials, as otherwise the
non-uniqueness of derivatives will break the result. Note that there is no
differentiability assumption in this result: this is made possible by our
choice of junk value $0$ for the derivative of a non-differentiable function.

\begin{rmk}
The formalization of the properties of $C^n$ functions contains many
curryfication or uncurryfication steps to identify for instance $\boL(E,
\boL(E, F))$ with continuous bilinear maps from $E \times E$ to $F$. In
textbooks such as~\cite{cartan_differential}, these identifications are
discussed at the beginning, and they are then used implicitly everywhere.
There is a subtlety here, though. Let $D(Df)$ denote the derivative of the
derivative of $f$, as an element of $\boL(E, \boL(E, F))$, and $D^2 f$ its
continuous bilinear map version. When we say that a function is $C^3$, are we
saying that $D(Df)$ has a derivative which is continuous, or that $D^2 f$ has
a derivative which is continuous? And do we get the same results by
uncurrying those two derivatives? The answer to these two questions is yes,
fortunately. This comes from the fact that the curryfication process is
itself a continuous linear bijection between two vector spaces of functions,
and continuous linear bijections respect differentiability and commute with
taking derivatives. With our choice of junk values that non-differentiable
functions have zero derivatives, composing with a continuous linear bijection
and taking the derivative commute even for non-differentiable functions,
which is extremely convenient in the formalization process. Note however that
this discussion is only true for functions on domains with unique
differentials: otherwise, the non-uniqueness of derivatives breaks this
commutation, which has to be accounted for in the formalization.
\end{rmk}

\section{Symmetry of iterated derivatives}
\label{sec:symm}

An important property of the second derivative in real vector space is its
symmetry $D^2 f(x) (v, w) = D^2 f(x) (w, v)$. Here is a precise version with
weak assumptions, see~\cite[Theorem~5.1.1]{cartan_differential}:
\begin{thm}
Let $E,F$ be two real normed vector spaces and let $f: E \to F$. Assume that
$f$ is differentiable on a neighborhood of a point $x$, and that $Df$ is
differentiable at $x$. Then $D^2 f(x) (v, w) = D^2 f(x) (w, v)$ for all
vectors $v, w$.
\end{thm}
This result follows from the fact that $(f(x+ tv+tw) -f(x+tv) -f(x+tw) +
f(x))/t^2$ converges to $D^2 f(x) (v, w)$ as $t$ tends to $0$. Since the
expression on the left is symmetric in $v,w$, it follows that the limit also
is.

This result is fundamental for many constructions in geometry. For instance,
define the Lie bracket of two vector fields $V, W : E \to E$ by
\begin{equation*}
  [V, W] (x) = DW(x) (V(x)) - DV(x) (W(x)).
\end{equation*}
It encodes the commutation default between the flows of $V$ and $W$. This
definition is invariant under local diffeomorphisms, hence it extends to
manifolds. To check the invariance, an order $2$ term comes out, of the form
\begin{equation}
\label{eq:cancel_symm}
  D^2 f(x) (V(x), W(x)) - D^2 f(x) (W(x), V(x)).
\end{equation}
It vanishes thanks to the symmetry of the second derivative. In the same way,
the relation $\dd \circ \dd = 0$ on differential forms ultimately boils down
to a computation relying on the symmetry of the iterated derivative.

We have formalized the symmetry of the second derivative, in a slightly more
general version as we allow convex domains:
\begin{lstlisting}
theorem Convex.second_derivative_within_at_symmetric
    (s_conv : Convex ℝ s) (hne : (interior s).Nonempty)
    (hf : ∀ x ∈ interior s, HasFDerivAt f (f' x) x) (xs : x ∈ s)
    (hx : HasFDerivWithinAt f' f'' (interior s) x) (v w : E) :
    f'' v w = f'' w v
\end{lstlisting}

This extends readily from real vector spaces to complex vector spaces, since
one can see a complex vector space as a real vector space. To cover also
$p$-adic manifolds, we would like to extend this property to other normed
fields, but unfortunately the symmetry of the second derivative fails in
general: the proof over the reals uses the mean value inequality, which does
not hold in discontinuous settings. We will describe a counterexample, which
is essentially~\cite[Corollary~84.2]{ultrametric_calculus}. The following
proposition shows that the derivative of a function on $\Q_p^2$ can be
essentially arbitrary, from which one can get non-symmetric second
derivatives. For simplicity, we work in the open subset $\Z_p$ of $\Q_p$,
where $p$ is any fixed prime number. The reader not familiar with $p$-adic
numbers may safely skip this argument.
\begin{prop}
\label{prop:non_symm} Let $p$ be a prime number. For $x, y \in \Z_p^2$, let
$M (x, y)$ be a matrix from $\Q_p^2$ to $\Q_p$ depending continuously on
$(x,y)$. Then there exists $f : \Z_p^2 \to \Q_p$ a function of class $C^1$
whose derivative is $M$.
\end{prop}
Note that the analogous result is completely false on $\R$, since the
integral of the derivative along a loop from $0$ to $0$ must be zero, so
everything is much more rigid.

The above proposition allows us to construct a counterexample to the symmetry
of second derivatives. Let $M(x, y) = (0, x)$. Then the corresponding
function $f$ has $M$ as its derivative, which is $C^\infty$, so it is itself
$C^\infty$. Furthermore, $\partial f/\partial x = 0$, so $\partial^2
f/\partial y\partial x = 0$, while $\partial f/\partial y = x$, so
$\partial^2 f/\partial x\partial y=1$.

\begin{proof}
A point in $\Z_p$ is determined by its successive reductions in $\Z/p^n \Z$:
at each stage, we refine by adding the data of the $n+1$-th digit.
Equivalently, $\Z_p$ is a union of $p^n$ balls of radius $p^{-n}$, each ball
being the union of $p$ balls of radius $p^{-(n+1)}$. The same applies in
$\Z_p^2$, except that there are $p^{2n}$ balls of radius $p^{-n}$. Let us
choose a center for each ball (the canonical choice is to take the point of
the ball whose subsequent digits are all zero), and let $C_n$ be the set of
centers of balls of depth $n$, with $C_n \subseteq C_{n+1}$.

We define $f$ recursively on $C_n$, starting from $f(0) = 0$. Suppose $f$ is
defined on $C_n$. For $b\in C_{n+1}$, let $a\in C_n$ be in the same ball of
radius $p^{-n}$. We then set $f(b) = f(a) + M(a) \cdot (b-a)$. Note that for
$a=b$, both sides give the same formula, so this inductive definition is
consistent. Furthermore, $\norm{f(b) - f(a)} \leq K \norm{b-a} \leq K p
^{-n}$. The summability of this series ensures that the function $f$ thus
defined on the increasing union of the $C_n$ is uniformly continuous.
Therefore, it extends uniquely to a function on $\Z_p^2$, which we again
denote by $f$.

It remains to be shown that $f$ is differentiable at each point $a$, with
derivative $M(a)$. Let $a_n \in C_n$ be the sequence of centers of the balls
containing $a$, and similarly for a point $b$. Then
\begin{equation*}
  f(b) - f(a) = \sum M(b_n) \cdot (b_{n+1} - b_n) - \sum M(a_n) (a_{n+1}-a_n).
\end{equation*}
Suppose that $\norm{b-a} = p^{-(N+1)}$. Then $a_n = b_n$ for $n\le N$, so
that we have
\begin{equation*}
  f(b) - f(a) = \sum_{n \ge N} M(b_n) \cdot (b_{n+1} - b_n) - \sum_{n\ge N} M(a_n) (a_{n+1}-a_n).
\end{equation*}
Let $\epsilon>0$. If $N$ is large enough, the two matrices $M(b_n)$ and
$M(a_n)$ are $\epsilon$-close to $M(a)$, by continuity of $M$. Therefore,
$f(b) - f(a)$ equals
\begin{align*}
\\& M(a) \sum_{n\ge N} (b_{n+1} - b_n) - M(a) \sum_{n\ge N} (a_{n+1}-a_n) \pm 2\epsilon \sum_{n\ge N} p^{-n}
  \\& = M(a) (b- b_N) - M(a) (a-a_N) \pm 2 \epsilon p^{-N} / (1-p^{-1})
  \\& = M(a) (b - a) \pm 2 \epsilon \norm{b-a} p / (1-p^{-1}).
\end{align*}
We have shown that, as soon as $b$ is close enough to $a$, we have for $C =
2p / (1-p^{-1})$ the estimate
\begin{equation*}
  \norm{f(b) - f(a) - M(a) (b-a)} \leq C \epsilon \norm{b-a}.
\end{equation*}
This is the definition of differentiability at $a$ with derivative $M(a)$.
\end{proof}

The construction is sufficiently explicit that we can write everything down.
For the counterexample to the symmetry of second derivatives, with $M(x,y) =
(0,x)$, the function $f$ reads
\begin{equation*}
  f \pare*{\sum x_n p^n, \sum y_n p^n} =\sum_{k<l} x_k y_l p^{k+l}.
\end{equation*}

The above pathology can not happen for analytic maps: if $f$ is locally given
by a series $f(x + y) = \sum_n p_n(y, \dotsc, y)$ as in
Definition~\ref{defn:analytic}, then one can compute the iterated derivative
of $f$ as follows: one has
\begin{equation}
\label{eq:deriv_analytic}
  D^n f(x) (v_1, \dotsc, v_n) = \sum_{\sigma \in \mathfrak{S}_n} p_n(v_{\sigma(1)}, \dotsc, v_{\sigma(n)})
\end{equation}
where $\mathfrak{S}_n$ denotes the set of permutations of $\{1, \dotsc, n\}$.
This expression is obviously symmetric in the $v_i$. In particular,
\begin{equation*}
D^2 f(x) (v_1, v_2) = p_2(v_1, v_2) + p_2(v_2, v_1).
\end{equation*}

\medskip

It follows from this discussion that symmetry of second derivatives of $C^n$
functions holds when $n \ge 2$ if the field is $\R$ or $\C$, and for
$n=\omega$ otherwise. All advanced calculus on vector fields or differential
forms is only possible under such smoothness conditions. This is the reason
why, over fields different from $\R$ or $\C$, Bourbaki only considers
analytic manifolds in~\cite{bourbaki_diff}.

We formalize the above formula for the iterated derivative of an analytic
function:
\begin{lstlisting}
theorem HasFPowerSeriesWithinOnBall.iteratedFDerivWithin_eq_sum
    (h : HasFPowerSeriesWithinOnBall f p s x r)
    (h' : AnalyticOn 𝕜 f s)
    (hs : UniqueDiffOn 𝕜 s) (hx : x ∈ s) (v : Fin n → E) :
    iteratedFDerivWithin 𝕜 n f s x v
      = ∑ σ : Perm (Fin n), p n (fun i ↦ v (σ i))
\end{lstlisting}

Let us now state that the second derivative is symmetric under a smoothness
condition that depends on the field. Let \lstinline{minSmoothness 𝕜 n} be
equal to $n$ if the field is isomorphic as a normed field to $\R$ or $\C$,
and $\omega$ otherwise. The general formalized version of the second
derivative symmetry is the following:
\begin{lstlisting}
theorem ContDiffAt.isSymmSndFDerivAt (v w : E)
    (hf : ContDiffAt 𝕜 n f x) (hn : minSmoothness 𝕜 2 ≤ n) :
    fderiv 𝕜 (fderiv 𝕜 f) x v w = fderiv 𝕜 (fderiv 𝕜 f) x w v
\end{lstlisting}

This design choice with a function \lstinline{minSmoothness} makes it
possible to give statements in a unified way. For instance, the fact that the
pullback of vector fields commutes with the Lie bracket reads as follows:
\begin{lstlisting}
lemma pullback_lieBracket
    (hn : minSmoothness 𝕜 2 ≤ n) (h'f : ContDiffAt 𝕜 n f x)
    (hV : DifferentiableAt 𝕜 V (f x))
    (hW : DifferentiableAt 𝕜 W (f x)) :
    pullback 𝕜 f (lieBracket 𝕜 V W) x =
      lieBracket 𝕜 (pullback 𝕜 f V) (pullback 𝕜 f W) x
\end{lstlisting}
This statement is proved first without the \lstinline{hn} assumption,
requiring instead that the second derivative of $f$ at $x$ is symmetric to
cancel out the terms~\eqref{eq:cancel_symm}. The theorem
\lstinline{ContDiffAt.isSymmSndFDerivAt} is then applied to show that the
symmetry is satisfied under \lstinline{hn}.

As an example of application, we can construct the Lie algebra of Lie groups
over arbitrary fields. A Lie group $G$ is a manifold with a multiplication
and an inversion which are $C^n$ for some $n$. Given two elements $v$ and $w$
of the tangent space $\frakg$ at the identity, one can form the corresponding
left-invariant vector fields $V$ and $W$ by transporting the vectors
everywhere in the group by left-multiplication. The value at the identity of
the Lie bracket $[V,W]$ of these vector fields is a new vector in $\frakg$,
called the (group) Lie bracket of $v$ and $w$. This endows $\frakg$ with a
Lie algebra structure, as the Lie bracket of vector fields satisfies the
required properties (linearity, antisymmetry, Jacobi identity). All these
properties are checked by direct computations in vector spaces, and
transported to the manifold since the pullback of the Lie bracket of vector
fields under local diffeomorphisms is well defined as explained above. This
requires some amount of smoothness, depending on the base field: over $\R$ or
$\C$, we need $C^3$ for this construction, while over other fields $C^\omega$
is needed. The formalized version is as follows:
\begin{lstlisting}
instance instLieAlgebraGroupLieAlgebra
    {I : ModelWithCorners 𝕜 E H} [TopologicalSpace G]
    [ChartedSpace H G] [Group G]
    [LieGroup I (minSmoothness 𝕜 3) G] [CompleteSpace E] :
    LieAlgebra 𝕜 (GroupLieAlgebra I G)
\end{lstlisting}
The details of the manifold or Lie group implementations can be safely
ignored by the reader here. The main point in this instance is the assumption %
\lstinline{LieGroup I (minSmoothness 𝕜 3) G} requiring that multiplication and
inversion are $C^3$ if the field is $\R$ or $\C$ and analytic otherwise, in a
unified way.

\begin{rmk}
There are other approaches to the Lie algebra of a Lie group. For instance,
it is straightforward to construct a Lie algebra structure on the space of
left-invariant derivations, i.e., maps $D$ from $C^n$ functions to $\R$
verifying $D(fg) = f(1) D(g) + g(1) D(f)$): one sets $[D,D'] f = D(D' f) -
D'(D f)$, and it follows from the derivation property that this Lie bracket
is again a derivation. In $C^\infty$ finite-dimensional real manifolds, there
is a canonical identification between the tangent space and the space of
derivations on $C^\infty$ functions. Therefore, in this specific case, this
constructs the Lie algebra structure on the tangent space at the identity of
a Lie group. This is the approach taken in~\cite{isabelle_liegroup}. However,
it fails over other fields, or for smoothness classes other than $C^\infty$,
as the identification between the tangent space and the space of derivations
fails in this generality. This approach is therefore not suitable for the
generality level of \mathlib.
\end{rmk}

\section{Conclusion}

We have described the formalization of the scale of smooth functions in the
Lean mathematical library \mathlib. It defines $C^n$ functions for $n \in \N
\cup \{\infty, \omega\}$, and departs in several ways from standard textbooks
definition. The design is motivated by the applications, notably to various
kinds of manifolds, that should be covered by the definition. The three main
constraints are that we should allow $C^n$ functions on domains, with respect
to fields which can be more general than $\R$ or $\C$, with a hierarchy
containing analytic functions. We have shown through several mathematical
examples that naive definitions taking these constraints into account would
fail natural properties that a good scale of $C^n$ functions should have,
forcing us to have a more elaborate definition.

The current definition is the result of a refactoring process that has been
ongoing for several years: whenever an application could not be implemented,
the definition had to be modified. The integration of analytic functions into
the hierarchy of $C^n$ functions, in particular, is motivated by the desire
to be able to formalize Lie groups over arbitrary fields, with regularity at
least $C^3$ over $\R$ or $\C$ and with analytic regularity otherwise. This
application has been completely formalized, which attests to the
effectiveness of the current definition, particularly in terms of providing
unified formulations of definitions and results across the entire smoothness
class.

We have not covered in this text the applications of $C^n$ functions to
analysis, although they are also very important. Let us just mention that the
Schwartz space of superpolynomially decaying $C^\infty$ functions is endowed
in \mathlib with its locally convex topological vector space structure by
constructing seminorms involving the norms of the iterated derivatives at
various points. Moreover, we have formalized the fact that the Fourier
transform is an isomorphism of the Schwartz space, and the Fourier inversion
formula. These results rely on many of the properties of $C^n$ functions
described in this paper.

\bibliography{biblio}

\newcommand{\etalchar}[1]{$^{#1}$}
\def\cprime{$'$}
\providecommand{\bysame}{\leavevmode\hbox to3em{\hrulefill}\thinspace}
\providecommand{\MR}{\relax\ifhmode\unskip\space\fi MR }
\providecommand{\MRhref}[2]{%
  \href{http://www.ams.org/mathscinet-getitem?mr=#1}{#2}
}
\providecommand{\href}[2]{#2}
\begin{thebibliography}{dMKA{\etalchar{+}}15}

\bibitem[ACR18]{Affeldt_Cohen_Rouhling_2018}
Reynald Affeldt, Cyril Cohen, and Damien Rouhling, \emph{Formalization
  techniques for asymptotic reasoning in classical analysis}, Journal of
  Formalized Reasoning \textbf{11} (2018), no.~1, 43--76.

\bibitem[Bou71]{bourbaki_diff}
Nicolas Bourbaki, \emph{\'el\'ements de math\'ematique. {F}asc. {XXXVI}.
  {V}ari\'et\'es diff\'erentielles et analytiques. {F}ascicule de r\'esultats
  ({P}aragraphes 8 \`a{} 15)}, Actualit\'es Scientifiques et Industrielles
  [Current Scientific and Industrial Topics], vol. No. 1347, Hermann, Paris,
  1971. \MR{281115}

\bibitem[Car71]{cartan_differential}
Henri Cartan, \emph{Differential calculus}, Hermann, Paris; Houghton Mifflin
  Co., Boston, MA, 1971, Exercises by C. Buttin, F. Rideau and J. L. Verley,
  Translated from the French. \MR{344032}

\bibitem[dMKA{\etalchar{+}}15]{demoura_lean}
Leonardo de~Moura, Soonho Kong, Jeremy Avigad, Floris van Doorn, and Jakob von
  Raumer, \emph{The {L}ean theorem prover (system description)}, Automated
  deduction---{CADE} 25, Lecture Notes in Comput. Sci., vol. 9195, Springer,
  Cham, 2015, pp.~378--388. \MR{3446905}

\bibitem[HIH13]{analysis_HOL}
Johannes H\"{o}lzl, Fabian Immler, and Brian Huffman, \emph{Type classes and
  filters for mathematical analysis in {I}sabelle/{HOL}}, Interactive theorem
  proving, Lecture Notes in Comput. Sci., vol. 7998, Springer, Heidelberg,
  2013, pp.~279--294. \MR{3111278}

\bibitem[IZ19]{immler_zhang}
Fabian Immler and Bohua Zhan, \emph{Smooth manifolds and types to sets for
  linear algebra in {I}sabelle/{HOL}}, CPP 2019, Association for Computing
  Machinery, 2019, pp.~65--77.

\bibitem[{mat}20]{mathlib}
{mathlib community}, \emph{The {L}ean mathematical library}, Proceedings of the
  9th {ACM} {SIGPLAN} International Conference on Certified Programs and
  Proofs, {CPP} 2020, 2020, pp.~367--381.

\bibitem[Muj86]{mujica_infinite_dim_complex_analysis}
Jorge Mujica, \emph{Complex analysis in {B}anach spaces}, North-Holland
  Mathematics Studies, vol. 120, North-Holland Publishing Co., Amsterdam, 1986.
  \MR{MR842435}

\bibitem[P{\"o}s21]{poeschel}
J{\"u}rgen P{\"o}schel, \emph{On the {S}iegel-{S}ternberg linearization
  theorem}, J. Dynam. Differential Equations \textbf{33} (2021), no.~3,
  1399--1425. \MR{4299973}

\bibitem[Sch84]{ultrametric_calculus}
Wilhelmus~H. Schikhof, \emph{Ultrametric calculus}, Cambridge Studies in
  Advanced Mathematics, vol.~4, Cambridge University Press, Cambridge, 1984, An
  introduction to $p$-adic analysis. \MR{791759}

\bibitem[SF25]{isabelle_liegroup}
Richard Schmoetten and Jacques~D. Fleuriot, \emph{Constructing the {L}ie
  algebra of smooth vector fields on a {L}ie group in {I}sabelle/{HOL}}, J.
  Automat. Reason. \textbf{69} (2025), no.~3, Paper No. 18, 29. \MR{4925296}

\end{thebibliography}
\bibliographystyle{amsalpha}
\end{document}